\documentclass[a4paper,twoside]{amsart}
\usepackage{amsmath}
\usepackage{amsfonts}
\usepackage{amssymb}
\usepackage{amsthm}
\usepackage{newlfont}
\usepackage{graphicx}
\usepackage{amscd}

\theoremstyle{plain}
\newtheorem{Th}{Theorem}[section]

\newtheorem{Prop}[Th]{Proposition}
\theoremstyle{definition}

\theoremstyle{remark}
\numberwithin{equation}{section}

\newcommand{\DD}{{\mathbb D}}

\newcommand{\ZZ}{{\mathbb Z}}

\newcommand{\bphi}{\boldsymbol{\phi}}

\begin{document}

\title[Non-commutative  $q$-Painlev\'{e} VI equation]
{Non-commutative  $q$-Painlev\'{e} VI equation}

\author{Adam Doliwa}

\address{Faculty of Mathematics and Computer Science, University of Warmia and Mazury in Olsztyn,
ul.~S{\l}oneczna~54, 10-710~Olsztyn, Poland}
\email{doliwa@matman.uwm.edu.pl}
\urladdr{http://wmii.uwm.edu.pl/~doliwa/}

\date{}

\begin{abstract} By applying suitable centrality condition to non-commutative non-isospectral lattice modified Gel'fand--Dikii type systems we obtain the corresponding non-autonomous equations. Then we derive non-commutative $q$-discrete Painlev\'{e} VI equation with full range of parameters as the $(2,2)$ similarity reduction of the non-commutative, non-isospectral and non-autonomous lattice modified Korteweg--de~Vries equation. We also comment on the fact that in making the analogous reduction starting from Schwarzian Korteweg-de~Vries equation no such "non-isospectral generalization" is needed.  
\end{abstract}
\maketitle

\section{Introduction}
The $q$-discrete Painlev\'{e} VI equation ($q$-$\mathrm{P_{VI}}$ ) is the following second order system
\begin{align} \label{eq:P-VI-1}
f_n f_{n+1} & = \frac{(g_n + t_n  c_1)(g_n + t_n  c_1^{-1})}
{(g_n + c_2)(g_n + c_2^{-1})}, \qquad t_{n+1} = \lambda t_n\\
\label{eq:P-VI-2}
g_n g_{n+1} & = \frac{(f_{n+1} + t_n \sqrt{\lambda} c_3)(f_{n+1} + t_n \sqrt{\lambda} c_3^{-1})}
{(f_{n+1} + c_4)(f_{n+1} + c_4^{-1})},
\end{align}
where $c_1, c_2, c_3, c_4$ are arbitrary parameters. It was obtained in~\cite{JimboSakai} on the basis of monodromy-preserving deformations of linear $q$-difference systems in close analogy to derivation~\cite{Fuchs} of the differential Painlev\'{e} VI equation. Also the singularity confinement property of $q$-$\mathrm{P_{VI}}$ equation was discussed there, and the continuous limit and special solutions in terms of $q$-hypergeometric functions were given as well. Actually, 
equations~\eqref{eq:P-VI-1}-\eqref{eq:P-VI-2} were written down in~\cite{RamaniGrammaticos} as the so called asymmetric discrete Painlev\'{e}~III equation. Bilinear structure and Schlesinger transforms of $q$-$\mathrm{P_{VI}}$ system were studied in \cite{JSRG}. In classification of Painlev\'{e} equations~\cite{Sakai} in terms of geometry of rational surfaces the $q$-$\mathrm{P_{VI}}$ system corresponds to the action of the affine Weyl group symmetry of type $D_5$. By quantizing the affine Weyl group action it was possible to construct in~\cite{Hasegawa} the quantization of the $q$-$\mathrm{P_{VI}}$ system.

According to Kruskal~\cite{GrammaticosRamani-review} 
Painlev\'{e} equations are located on a borderline between trivial integrability (usually linearisability) and non-integrability. The Painlev\'{e} equations, both differential and difference, have found numerous applications in theoretical physics and mathematics, for example in analysis of corelation functions of two dimensional Ising model \cite{WCTB} or one dimensional Bose gas \cite{JMMS}, two dimensional quantum gravity \cite{BrezinKazakov}, random matrices \cite{TracyWidom}, integrable reductions of the Einstein equations \cite{Tod} or distribution of zeros of Riemann's $\zeta$-function \cite{ForresterOdlyzko}. It is well known that all the six differential Painlev\'{e} equations can be obtained as reductions of partial differential equations~\cite{AblowitzSegur,Adler,CGM}. However, in spite of various successful attempts
\cite{FWN-lB,GRSWC,KNY-qKP,NRGO,KakeiKikuchi,Ormerod,OKQ} still there does not exist such a procedure for all the discrete Painlev\'{e} equations as classified in~\cite{Sakai}. 

Non-commutative versions of integrable maps or discrete systems \cite{FWN-Capel,Kupershmidt,BobSur-nc,Nimmo-NCKP,DF-K,Dol-GD} are of growing interest in mathematical physics. They may be considered as a useful platform to more thorough understanding of integrable quantum or statistical mechanics lattice systems, where the quantum Yang--Baxter equation~\cite{Baxter,QISM} or its multidimensional generalization~\cite{Zamolodchikov} play a role. In studying non-commutative versions of known integrable systems one tries to capture similarities in structures relevant to integrability of the corresponding both classical and quantum models. In analogy to constructing discrete versions of integrable differential equations this can be considered not only as broadening the scope of integrable systems theory or finding new areas of its applications, but also one aims to provide its deeper understanding. Also non-commutative versions of Painlev\'{e} equations have attracted attention recently~\cite{RR-ncPII,BertolaCafasso-CMP,CafassoIglesia}. From the fully non-commutative perspective known quantum Painlev\'{e} equations~\cite{NGR} can be considered as reductions, obtained due to a particular commutation rule consistent with the evolution, what is a dual approach to the standard (canonical) quantization problem understood as a deformation of Poisson algebras. This point of view allowed recently to derive, in the context of integrable discrete models~\cite{Sergeev-ultralocality,DoliwaSergeev-pentagon}, certain important commutation relations starting from the ultralocality principle. 

Recently a non-isospectral non-commutative lattice analog of the modified Gel'fand--Dikii hierarchy was obtained~\cite{Dol-GD} from periodic reductions of the Desargues maps~\cite{Dol-Des} which provide geometric meaning to the non-commutative discrete Kadomtsev--Petviashvili (KP) system~\cite{Nimmo-NCKP}. It was conjectured there, see also~\cite{LeviRagniscoRodriguez},  
that the presence of the non-isospectral factor should provide in the dimensional reduction process additional parameters and would give rise to discrete Painlev\'{e} equations in their full generality. In the present work, following analogous derivation  \cite{Ormerod} from (commutative and isospectral) non-autonomous lattice modified Korteweg--de~Vries (mKdV) equation, we demonstrate such a process on the example of $q\mathrm{-P_{VI}}$ equation on the non-commutative level. 

The paper is constructed as follows. In Section \ref{sec:GD} we obtain novel non-commutative non-isospectral non-autonomous (but with central factors) lattice equations of the modified Gel'fand--Dikii type. Then in Section \ref{sec:P-VI} we impose $(2,2)$ similarity reduction condition on the simplest such equation (with period two) to get the non-commutative $q\mathrm{-P_{VI}}$ equation. In the Appendix we comment on the fact that in making the self-similarity reduction starting from Schwarzian Korteweg-de~Vries (SKdV) equation~\cite{OKQ,OKHQ} no "non-isospectral generalization" is needed. 

Given a function $f$ from $N$ dimensional integer lattice $\ZZ^N$, we write $f_{(i)}(n_1, \dots , n_i, \dots , n_N)$ instead of 
$f(n_1, \dots , n_i+1, \dots , n_N)$, and we skip usually the argument $n=(n_1, n_2, \dots , n_N)$.

\section{Gel'fand--Dikii reductions of non-commutative discrete KP hierarchy} \label{sec:GD}
Consider~\cite{Dol-GD} the chain (labelled by $k\in\ZZ$) of linear equations
\begin{equation} \label{eq:lin-lKPh}
\bphi_{k+1} - \bphi_{k(i)} = \bphi_k u_{i,k}, \qquad i=1,\dots ,N,
\end{equation}
with the coefficients restricted by the compatibility relations~\cite{Dol-YB}
\begin{equation} \label{eq:KP-u-solved}
u_{i,k(j)} = ( u_{i,k} - u_{j,k})^{-1} u_{i,k} ( u_{i,k+1} - u_{j,k+1}), \qquad i\neq j.
\end{equation} 
As a consequence of equations \eqref{eq:KP-u-solved} there exist potentials $r_k$, $k\in\ZZ$, such that $u_{i,k} = r_k^{-1} r_{k(i)}$. 

The quasi-periodic reduction $\bphi_{k+L} = \bphi_k \mu_k$ of the chain, with the the monodromy factors 
$\mu_k\colon \ZZ^N\to\DD^\times$ restricted by the condition $\mu_{k+1} = \mu_{k(i)}$, $i=1,\dots ,N$, results in the constraints 
\begin{equation} \label{eq:u-per}
u_{i,k+L} = \mu_{k}^{-1} u_{i,k}\mu_{k(i)}, \qquad
r_{k+L} = r_k \mu_k.
\end{equation}
These give rise to the matrix linear problem
\begin{equation} \label{eq:Lm-kp-K}
\left(\bphi_1, \bphi_2, \dots , \bphi_L \right)_{(i)}= 
\left(\bphi_1, \bphi_2, \dots , \bphi_L \right)
\left(  \begin{array}{ccccc}  
-r_{1}^{-1} r_{1(i)} & 0  & \cdots & 0 & \mu \\
1 & -r_{2}^{-1} r_{2(i)}  & 0 & \hdots  & 0 \\
0 & 1 & \ddots &  &  \vdots \\
\vdots &  & & -r_{L-1}^{-1} r_{L-1(i)}  & 0 \\
0 & 0 & \ \hdots  & 1 & -r_{L}^{-1} r_{L(i)}  \end{array} \right) ,
\end{equation}
and the corresponding non-isospectral non-commutative discrete Gel'fand--Dikii type equations
\begin{align}
\label{eq:GD-K-mu}
(r_{k(j)}^{-1} - r_{k(i)}^{-1}) r_{k(ij)} & = r_{k+1}^{-1} (r_{k+1(i)} - r_{k+1(j)} ), \quad k=1,\dots ,L-1, \\ \nonumber
(r_{L(j)}^{-1} - r_{L(i)}^{-1}) r_{L(ij)} & = 
\mu^{-1} r_{1}^{-1} (r_{1(i)} - r_{1(j)} ) \mu_{(\sigma)} \qquad i\neq j,
\end{align}
where $\mu = \mu_1$ is an arbitrary function of the single variable $n_\sigma = n_1 + \dots + n_N$.

\begin{Prop}
Assume that the products $\mathcal{U}_{i} = u_{i,1} \dots u_{i, L} \mu_1^{-1} $
belong to the center of the division ring $\DD$, then by direct calculation using equations \eqref{eq:KP-u-solved}, \eqref{eq:u-per} and properties of the factor $\mu$  we show that 
(1) $\mathcal{U}_{i} = u_{i,k} \dots u_{i, k+L-1} \mu_k^{-1}$ for arbitrary index $k$;
(2) $\mathcal{U}_{i}$ is a function of the single variable $n_i$.
\end{Prop}
Under such centrality assumption we may reduce the number of dependent variables (and equations) by one paying the "price" of introducing arbitrary non-autonomous factors $\mathcal{U}_{i}$ replacing in the equations 
\begin{equation}
r_L^{-1} r_{L(i)} = r_{L-1 (i)}^{-1} r_{L-1} \dots r_{1 (i)}^{-1} r_{1} \, \mathcal{U}_{i} \mu, \qquad i=1, \dots ,N.
\end{equation}
In the simplest case $L=2$, when we denote $x=r_1$, this procedure leads to the linear problem
\begin{equation} \label{eq:Lm-mKdV}
\left(\bphi_1, \bphi_2 \right)_{(i)}= 
\left(\bphi_1, \bphi_2 \right)
\left(  \begin{array}{cc}  
-x^{-1} x_{(i)} & \mu \\
1 & -x_{(i)}^{-1} x \, \mathcal{U}_{i} \mu  \end{array} \right) ,
\end{equation}
of the non-isospectral and non-autonomous version of the lattice non-commutative mKdV (or Hirota) 
equations~\cite{BobSur-nc}
\begin{equation} \label{eq:Hirota-SG}
\left( x_{(j)}^{-1}  - x_{(i)}^{-1} \right) x_{(ij)} = 
\left( x_{(i)}^{-1} \mathcal{U}_i - x_{(j)}^{-1} \mathcal{U}_j \right) x \mu , \qquad i\neq j.
\end{equation}

\section{Similarity reduction to $q$-Painleve VI} \label{sec:P-VI}
In this Section we present the reduction procedure from the non-isospectral non-autonomous lattice non-commutative mKdV equations \eqref{eq:Hirota-SG} to a non-commutative version of the 
$q$-$\mathrm{P_{VI}}$ equation written in the form (obtained by suitable rescaling of the dependent variables and the time function)
 In performing the reduction procedure we follow recent work~\cite{Ormerod} where the 
$q\mathrm{-P_{VI}}$ equation \eqref{eq:P-VI-1}-\eqref{eq:P-VI-2} with some constraint imposed on its parameters was obtained. In our work it is the presence of the non-isospectral factor which allows to obtain the full $q\mathrm{-P_{VI}}$ equation.

In equation \eqref{eq:Hirota-SG} we take $N=2$, and we assume that $\mu$ takes values in the center of $\DD$. Let us impose the reduction condition $x_{(1122)} = x$, which is compatible with the equation if 
\begin{equation}
\mathcal{U}_{i(ii)} \mu_{(\sigma \sigma \sigma \sigma)} = \mathcal{U}_{i} \mu, \qquad i=1,2.
\end{equation}
By separation of variables there exist a non-zero central constant $q$ such that
\begin{align}
\mu(n_\sigma) & = \alpha_k q^{n_\sigma}, \qquad \quad k= n_\sigma \quad \mod 4,\\
\mathcal{U}_{i}(n_i) & = \beta_{i,k} q^{-2n_i}, \qquad k= n_i \quad \mod 2, \qquad i=1,2,
\end{align}
for certain non-zero parameters $\alpha_k$, $\beta_{i,k}$ are certain non-zero central parameters.

Supplement the half-period direction $\left( \begin{array}{r} 1 \\ 1 \end{array} \right)$ by a transversal vector 
$\left( \begin{array}{r} 0 \\ 1 \end{array} \right)$ to a basis of the $\ZZ^2$ lattice. 
The corresponding transversal variable $m$ and the half-periodicity variable $p$ 
\begin{equation*}
m\left( \begin{array}{r} 0 \\ 1 \end{array} \right) + p \left( \begin{array}{r} 1 \\ 1 \end{array} \right) =
\left( \begin{array}{r} n_1 \\ n_2 \end{array} \right)
\end{equation*}
can be expressed by the original variables as
\begin{equation}
m = n_2 - n_1 , \qquad p = n_1.
\end{equation}
The evolution variable $n$ will be the double shift in $m$. 

Let us fix $(n_1,n_2)$ and consider the pattern consisting of four points $w^0_n = x(n_1,n_2-1)$,  $w^1_n = x(n_1,n_2)$,  
$w^2_n = x(n_1+1,n_2)$ and  $w^3_n = x(n_1+1,n_2+1)$, see Figure~\ref{fig:pattern}. 
\begin{figure}
\begin{center}
\includegraphics[width=8cm]{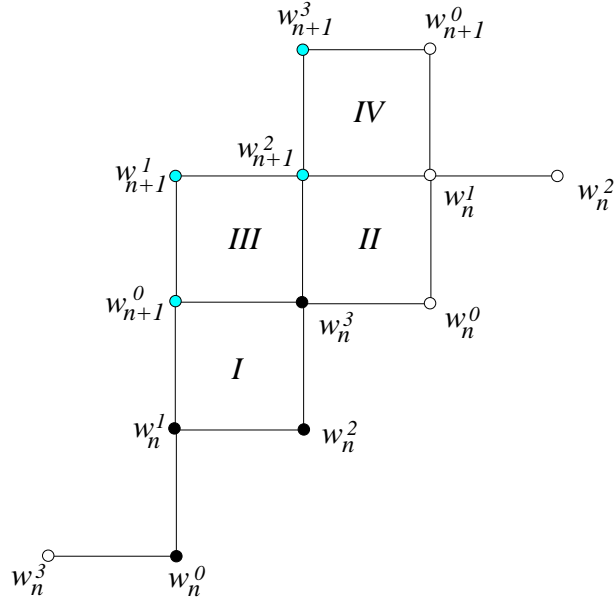}
\end{center}
\caption{The repeating pattern for $q$-$\mathrm{P_{VI}}$ and four elementary quadrilaterals used to construct the system} 
\label{fig:pattern}
\end{figure}
On elementary quadrilateral (plaquette) $I$ with the lower left corner function $x(n_1,n_2)$ we find the relation
\begin{equation}
w^0_{n+1} = (w^3_n + \mathcal{U}_2 \mu w^1_n) ( w^3_n + \mathcal{U}_1 \mu w^1_n)^{-1} w^2_n.
\end{equation}
Similarly, on plaquette $I\!I$ with the lower left corner function $x(n_1+1,n_2+1)$ we have
\begin{equation}
w^2_{n+1} = (w^1_n + \mathcal{U}_{2(2)} \mu_{(\sigma\sigma)} w^3_n) ( w^1_n + \mathcal{U}_{1(1)} 
\mu_{(\sigma\sigma)} w^3_n)^{-1} w^0_n,
\end{equation}
which combined with the previous equation gives
\begin{equation} \label{eq:P6-1}
w^0_{n+1} (w^2_{n+1})^{-1} =  
\frac{ 1 +  \mathcal{U}_2 \mu \, w^1_n (w^3_n)^{-1} }{ 1 +  \mathcal{U}_1 \mu \, w^1_n (w^3_n)^{-1} }
w^2_n (w^0_n)^{-1}  
\frac{ w^1_n (w^3_n)^{-1}  +  \mathcal{U}_{1(1)} \mu_{(\sigma\sigma)} }{  w^1_n (w^3_n)^{-1}  +
(\mathcal{U}_{2(2)} \mu_{(\sigma\sigma)} } .
\end{equation} 
Similar considerations for plaquette $I\!I\!I$ with the lower left corner function $x(n_1,n_2+2)$, and
for plaquette $I\! V $ with the lower left corner function $x(n_1+1,n_2+2)$, give 
\begin{equation} \label{eq:P6-2}
w^1_{n+1} (w^3_{n+1})^{-1} =  
\frac{ 1 +  \mathcal{U}_{2(2)} \mu_{(\sigma)} \, w^0_{n+1} (w^2_{n+1})^{-1} }
{ 1 +  \mathcal{U}_1 \mu_{(\sigma)}  \, w^0_{n+1} (w^2_{n+1})^{-1} }
w^3_n (w^1_n)^{-1}  
\frac{ w^0_{n+1} (w^2_{n+1})^{-1}  +  \mathcal{U}_{1(1)} \mu_{(\sigma\sigma\sigma)} }{  w^0_{n+1} (w^2_{n+1})^{-1}  +
(\mathcal{U}_{2(22)} \mu_{(\sigma\sigma\sigma)} } .
\end{equation} 

Define 
\begin{equation}
f_n = \gamma w^0_n (w^2_n)^{-1} \mathcal{U}_1(n_1) \mu(n_\sigma), \qquad
g_n = \delta w^1_n (w^3_n)^{-1} \mathcal{U}_2(n_2) \mu(n_\sigma),
\end{equation}
where $\gamma$ and $\delta$ are constants to be determined to match the canonical form of $q\mathrm{-P_{VI}}$. Moreover notice, that shift in $n$ by $1$ exchanges 
$\alpha_k$ with $\alpha_{k+2}$ and $\alpha_{k+1}$ with $\alpha_{k+3}$, therefore not to have to introduce "asymmetric" form of 
$q\mathrm{-P_{VI}}$ we need to identify elements of each pair. To obtain the final form of the equation we also assume (without loosing generality) that both $n_1$ and $n_2$ are even. Then for
\begin{equation}
\gamma = \frac{1}{q \alpha_0 \sqrt{\beta_{1,0} \beta_{2,0}}}, \qquad
\delta = \frac{1}{\alpha_0 \sqrt{\beta_{1,0} \beta_{2,1}}}\end{equation} 
we obtain the non-commutative $q$-$\mathrm{P_{VI}}$ system
\begin{align} \label{eq:nc-P-VI-1}
f_{n+1} & = \frac{ g_n + t_n  c_1^{-1} }{g_n + c_2^{-1}} f_n^{-1} \frac{g_n + t_n  c_1} {g_n + c_2} , \qquad 
t_{n+1} = \lambda t_n, \\
\label{eq:nc-P-VI-2}
g_{n+1} & = \frac{f_{n+1} + t_n \sqrt{\lambda} c_3^{-1}}{f_{n+1} + c_4^{-1}}  g_n^{-1} 
\frac{f_{n+1} + t_n \sqrt{\lambda} c_3 }{f_{n+1} + c_4}
\end{align}
with the time-function in the form
\begin{equation}
t_n = t_0 \lambda^n, \qquad \lambda = q^4, \qquad t_0 = \sqrt{\frac{\beta_{1,0} \beta_{1,1} }{ \beta_{2,0} \beta_{2,1} }},
\end{equation}
and the parameters of the equation given by
\begin{equation*}
c_1 = \alpha_0 \sqrt{\beta_{1,1} \beta_{2,0}}, \qquad
c_2  = \alpha_0 \sqrt{\beta_{1,0} \beta_{2,1}}, \qquad
c_3  = \alpha_{1} \sqrt{\beta_{1,1} \beta_{2,1}}, \qquad
c_4  = \alpha_{1} \sqrt{\beta_{1,0} \beta_{2,0}}.
\end{equation*}
The "reverse" transformation depends on two arbitrary parameters $u,v$ (scale factors) and reads
\begin{gather*}
\alpha_0 = u \sqrt{c_1 c_2}, \qquad \alpha_1 = u \sqrt{c_3 c_4},\\
\beta_{1,0} = \frac{1}{uv} \sqrt{\frac{c_2 c_4}{c_1 c_3}}, \qquad \beta_{1,1} = \frac{1}{uv} \sqrt{\frac{c_1 c_3}{c_2 c_4}}, \qquad
\beta_{2,0} = \frac{v}{u} \sqrt{\frac{c_1 c_4}{c_2 c_3}}, \qquad \beta_{2,1} = \frac{v}{u} \sqrt{\frac{c_2 c_3}{c_1 c_4}}. 
\end{gather*}

\section{Concluding remarks}
Making use of the non-isospectral generalization of the non-autonomous lattice mKdV equation in the course of the $(2,2)$ similarity reduction we obtained discrete $q\mathrm{-P_{VI}}$ equation in its full generality. Actually, we proposed a non-commutative version of the equation. In doing that by application of the appropriate centrality condition~\cite{Dol-YB} we first derived appropriate non-commutative generalization of the non-isospectral non-autonomous lattice modified Gel'fand--Dikii type equations~\cite{Dol-GD}. We expect to obtain in a similar way other more complicated $q$-Painleve equations together with their non-commutative analogs.   

\appendix
\section{Disappearance of the non-isospectral factor in transition to the Schwarzian form}
In recent works~\cite{OKQ,OKHQ} similarity reduction of the lattice SKdV equation were studied, and in particular the most general form of the $q\mathrm{-P_{VI}}$ equation has been obtained in \cite{OKHQ} by applying (Moebius) twisted reduction. Remarkably no "non-isospectral generalization" was needed. In this Section we would like to comment on this fact showing how in the transition to the Schwarzian form the non-isospectral factor disappears. Below we discuss the commutative case only. 

\begin{Prop}
Given solution 
\begin{equation}
\left( \begin{array}{cc} \bphi_1 & \bphi_2 \end{array} \right) = 
\left( \begin{array}{cc} \phi_1^0 & \phi_2^0  \\ \phi_1^1 & \phi_2^1 \end{array} \right)
\end{equation}
to the linear problem \eqref{eq:Lm-mKdV} of the non-isospectral non-autonomous lattice mKdV equation
\eqref{eq:Hirota-SG} then the functions $\psi_k = \phi^{1}_k/\phi^0_k$, $k=1,2$ satisfy the following nonlinear equations
\begin{equation}
\frac{(\psi_{(ij)} - \psi_{(j)} )(\psi_{(i)} - \psi )}{(\psi_{(ij)} - \psi_{(i)} )(\psi_{(j)} - \psi )}
= \frac{\mathcal{U}_j -1}{\mathcal{U}_i -1}, \qquad i\neq j,
\end{equation}
known as the lattice non-autonomous SKdV system~\cite{NCQ-SKdV}.
\end{Prop}
\begin{proof}
We will demonstrate the result for $\psi = \psi_1$ only, for the second function the reasoning is analogous.
From the linear problem we have
\begin{equation*}
\psi_{(i)} - \psi = \frac{ \phi_2^1 \phi_1^0 - \phi_1^1 \phi_2^0}{\phi_{1(i)}^0 \phi_1^0},
\end{equation*}
which gives 
\begin{equation*}
\frac{(\psi_{(ij)} - \psi_{(j)} )(\psi_{(i)} - \psi )}{(\psi_{(ij)} - \psi_{(i)} )(\psi_{(j)} - \psi )} =
\frac{(\phi_2^1 \phi_1^0 - \phi_1^1 \phi_2^0)_{(j)}}{(\phi_2^1 \phi_1^0 - \phi_1^1 \phi_2^0)_{(i)}}, \qquad i\neq j.
\end{equation*}
After making use of the linear problem once again we obtain
\begin{equation*}
(\phi_2^1 \phi_1^0 - \phi_1^1 \phi_2^0)_{(i)} = \mu (\phi_2^1 \phi_1^0 - \phi_1^1 \phi_2^0) 
(\mathcal{U}_i -1),
\end{equation*}
which implies the statement.
\end{proof}
Therefore the non-isospectral factor is removed from the equations, but it remains "inside" of the corresponding solution $\psi$ of SKdV equation. This fact can be explained on the geometric level of periodic reductions of Desargues maps studied in~\cite{Dol-GD} as follows. The periodicity condition when expressed in homogeneous coordinates $\bphi_k$ is equivalent to proportionality of the corresponding vectors. The functions $\psi$ play the role of non-homogeneous coordinates, where we have strict periodicity.

\section*{Acknowledgments}
The research was initiated during author's work at Institute of Mathematics of the Polish Academy of Sciences.
The paper was supported in part by Polish Ministry of Science and Higher Education grant No.~N~N202~174739.

\bibliographystyle{amsplain}

\providecommand{\bysame}{\leavevmode\hbox to3em{\hrulefill}\thinspace}

\end{document}